\documentclass[onetabnum,letterpaper]{siamltex}
\usepackage{amsfonts,amsmath}
\usepackage{subfig}
\usepackage{MnSymbol}
\usepackage{tikz}
\usepackage{bm}
\usepackage{marginnote}
\usepackage{cases}
\usetikzlibrary{shapes,patterns,snakes}
\usetikzlibrary{positioning,fit,calc}

\def\d{\displaystyle}

\def\R{\mathbb{R}}
\def\taubar{\overline{\tau}}

\def\begeq{\begin{equation}}
\def\endeq{\end{equation}}
\def\bc{\begin{center}}
\def\ec{\end{center}}

\title{On the 
mean square displacement in 
L\'evy walks}

\author{Christoph B\"{o}rgers\footnotemark[2] and Claude Greengard\footnotemark[3]}

\begin{document}
\maketitle

\renewcommand{\thefootnote}{\fnsymbol{footnote}}
\footnotetext[2]{Department of Mathematics, Tufts University,  Medford, MA 02155. 
email: cborgers@tufts.edu}
\footnotetext[3]{Two Sigma Investments, LP, and Courant Institute of Mathematical Sciences, New York University, New York, NY 10012. email:cg3197@nyu.edu}
\renewcommand{\thefootnote}{\arabic{footnote}}

\begin{abstract}
Many physical and biological processes are modeled by ``particles" undergoing L\'evy random walks. A feature of significant interest in these systems is the mean square displacement (MSD) of the particles. Long-time asymptotic approximations of the MSD have been established, via the Tauberian Theorem, for systems in which the distribution of the step durations is asymptotically a power law of infinite variance. We extend these results, using elementary analysis, and obtain closed-form expressions as well as power law bounds for the MSD in equilibrium, and representations of the MSD as sums of super-linear, linear, and sub-linear terms. We show that the super-linear components are determined by the mean and asymptotics of the step durations, but that the linear and sub-linear components (whose size has implications for the accuracy of the asymptotic approximation) depend on the entire distribution function. 
\end{abstract}



\pagestyle{myheadings}
\thispagestyle{plain}
\markboth{C.\ B\"{o}rgers and C.\ Greengard}{Mean square displacement in L\'evy walks}

\section{Introduction} We consider a random walker in $\R^d$ moving at random velocities  
in straight line segments of random duration. At the end of each flight segment, a new velocity and duration are chosen at random. We assume that the durations and velocities are independent. Denote by $F$ the distribution function of the duration. When $F$ has infinite variance, the case of primary focus in this
paper, this process is called a L\'evy walk \cite{klafter1987stochastic, klafter_sokolov,Shlesinger_2001,Shlesinger_et_al, Zaburdaev2015levy}. A swarm of 
random walkers starting at the same location and performing independent L\'evy walks exhibits super-diffusion, 
that is, faster than linear growth in the {\em mean square displacement} (MSD) --- the mean of the square of their distances from the origin.
One finds  many discussions of 
physical and biological applications of L\'evy walks and related
random processes leading to super-diffusion in the literature. Our own 
interest in the subject started when we modeled the evolution of a rarefied gas in an 
ultra-thin planar channel as a succession of infinite-variance flights and showed that the molecular density approaches a Gaussian distribution with a variance which grows super-linearly  \cite{BGT,Dogbe,Mellet_et_al}. 
Other examples of
applications include certain kinds of transport in fluid flow \cite{Shlesinger_et_al,Zaks_et_al}, 
transport in biological cells \cite{Bressloff_book,Song_et_al}, the 
migration of bacteria
\cite{Bacteria_Levy_walk}, predator search behavior \cite{Sims:2008de}, and traveling humans \cite{Human_Travel}. 
Also, Lorentz gases \cite{Dolgopyat,Zarfaty:2018tz} and other billiards problems \cite{Balint_Dolgopyat}, in which there is no
randomness in a strict sense but, rather, deterministic chaos, have been approximated by super-diffusive
 random walks.

The most commonly studied L\'evy walks are those in which the speeds of the walkers are fixed. However, there are many interesting examples in which both the directions and the speeds of the flights are random \cite{zaburdaev2008random}. For example, in the gas flow between horizontal plates with Maxwellian reflection conditions mentioned above, individual molecules undergo random flights whose durations and horizontal velocities are independent, with durations of infinite variance and finite expected square speed \cite{BGT}. By allowing speeds to be random, but assuming that the expectation of the squared speed is finite, we include this and other examples without complicating the analysis beyond that of the single-speed case.

A fundamental quantity characterizing super-diffusion is
the MSD. The common approach to analyzing the MSD is to formulate
an equation describing the density of the walker location as a function of time and then to take a Fourier transform with respect to the space variables and a Laplace transform with respect to the time variable. One can then use the Tauberian theorem to deduce the long-time asymptotics of the MSD from the behavior of derivatives with respect to the Fourier variable near 0 
\cite{Geisel:1985ir, klafter_sokolov,Wang:1992gv}. 

An alternative and more elementary approach is to express the MSD as an integral over the velocity
auto-correlation, and evaluate or analyze the integral; see for instance \cite{Balint_Dolgopyat,dechant2014scaling}. 
We present a mathematical analysis of the basic properties of the MSD, using this approach, under the assumption that 
$F$ has finite expectation. 

While the standard approach is very useful in providing insights into various aspects of the {\em asymptotic} behavior of L\'evy walks, the more elementary approach taken in this paper allows the derivation of {\em exact},
explicit expressions, valid for all times, for the MSD, for important classes of problems (Section 4). This enables detailed understanding of the accuracy of 
asymptotic approximations for the MSD (Section 6). In other examples, when exact expressions cannot be obtained, our approach permits the derivation of {\em time-dependent   bounds}, valid for all times (Corollary 7.4).

We will distinguish between the ``equilibrium" MSD, $M_{eq}$, and the ``transitional'' MSD, $M_{tr}$. Precise definitions
are given in Section \ref{sec:background}. The difference between the two cases
lies in the interpretation of ``$t=0$". In the
equilibrium case, 
$t=0$ should be thought of as the time at which we start watching a random walk that has
been going on for a long (strictly speaking, infinite) time. In the transitional case, 
the walk begins at time $t=0$. The equilibrium case is easier to analyze since the MSD is then expressible as a simple convolution integral, and we focus on this case.

Long-time asymptotic approximations have been derived for various duration distributions, for both $M_{tr}$ and $M_{eq}$. These approximations depend only on the tails of the distributions. However, accuracy of the asymptotic approximation may only emerge
after an extremely long time. For instance, if $M_{eq}$ is proportional to $t \log t$ to leading order, as occurs in numerous applications
\cite{Balint_Dolgopyat,dettmann2014diffusion,Zarfaty:2018tz}, there is typically a correction term of order $t$. For this
term to become negligible, $1/\log t$ must become small, so $t$ must be extremely large. We show 
that the presence or absence of a linear correction term depends on the
entire distribution $F$, not just on its tail. Similarly we show that the presence of a
logarithmic factor in $1-F$ may cause $M_{eq}$
 to consist, to leading order, of a sum of
two terms that only differ by a factor proportional to $\log t$.
Highly precise and complete knowledge of $F$ may therefore be needed to 
ascertain whether leading-order long-time asymptotic approximations for $M_{eq}$ are accurate approximations
to the actual $M_{eq}$ at times of physical interest.

Of course, distributions found in real-world applications need not be exactly power laws asymptotically. We show that the model is robust in the sense that for distributions bounded by power laws, the MSDs also satisfy related bounds.

We relate the transitional MSD, $M_{tr}$, to the equilibrium MSD, $M_{eq}$.
Our results on $M_{tr}$ are weaker than those on $M_{eq}$ and mostly concern the leading-order asymptotic
 behavior, but we also give a result on correction terms for $M_{tr}$. Finally, we present the asymptotic MSD for free molecular flow in planar channels \cite{BGT}.

\section{Background and notation}
\label{sec:background}

Let  $T_i$, $i=0,1,2,\ldots$, be random variables
with 
$$
0 \leq T_0 < T_1 < T_2 < \ldots~ .
$$
We will consider a random walker changing velocities at times $T_i$. 
We assume that the increments $\tau_i=T_i-T_{i-1}$, $i=1,2,3,\ldots$, are independent and
identically distributed positive random variables. The equilibrium and transitional cases differ in the choice of distribution of $T_0$, as discussed below.

Let $F$ denote the distribution of the $\tau_i$. We assume throughout this paper that $F$ is continuous and that
its mean, $\taubar$, is finite:
$$
\taubar = \int_0^\infty (1-F(s))~\! ds < \infty.
$$
Our primary concern is this paper is with the case when the
 second moment of $F$,
\begin{equation}
\label{eq:sigma_in_terms_of_F}
\sigma^2 = \int_0^\infty 2s(1-F(s)) ~\!ds,
\end{equation}
is infinite.

At each time $T_i$ a new velocity vector, $V_i$, is chosen. For simplicity of exposition, we will often refer to the times $T_0,T_1,\dots$ as "collision" times, thinking of a particle colliding with a background. Thus  the random selection of a new $V_i$ is thought of as the 
result of a ``collision".  We assume that the velocities are identically distributed and independent of each other and of the collision times. We'll also assume that the mean velocity is zero, since if it were not, one could simply subtract the mean velocity and consider the shifted velocities (the only change would be a drift in the direction of the mean velocity). Denote by $v$ the characteristic speed,
$$
v=\sqrt{E\left[|V_i|^2\right]},
$$
where $| \cdot |$ denotes the Euclidean norm.
As we shall see below, given a duration distribution $F$, the MSD is proportional to $v^2=E[|V_i|^2]$.
This is the only dependence of the MSD on the velocity distribution. So, for example, given $F$, walks in $d$ dimensions with post-collision velocities uniformly distributed on the sphere of radius $v$ and those which go only in axial directions, with equal probabilities, at speed $v$, have identical MSDs --- not just asymptotically, but for all time.

We consider two choices of $T_0$:

\begin{enumerate}
\item 
{\em Transitional case:} $T_0=0$. 
\item  {\em Equilibrium case:} $T_0>0$ is random and independent of the $\tau_i$, with 
distribution function
\begin{equation}
\label{eq:def_F_equi}
F_0(t) = \frac{1}{\taubar} \int_0^t (1-F(s)) ds,  ~~~ t>0.
\end{equation}
\end{enumerate}
To review the significance of the distribution $F_0$, recall \cite{Feller_II}
that for any $t \geq 0$, 
$$
Z_t = \min\{T_j ~:~ j \geq 0,~~T_j>t\} - t > 0
$$
is called the {\em residual life at time $t$}. Denote
its distribution function by $F_R(t,z)$, $z > 0$, so that
for $t \geq 0$, $z>0$, 
\begin{equation*}
\label{eqF:alternative}
F_R(t,z) = 
P(\exists  j  \geq 0  ~~~t < T_j \leq t+z).
\end{equation*}
\begin{lemma} 
\label{lemma:F_R_continuous}
If $F$ is continuous, then 
$F_R(t,z)$ is a continuous function of $(t,z) \in [0,\infty) \times (0,\infty)$.
\end{lemma} 
\begin{proof} Let $t \geq 0$, $z > 0$. Let $\Delta t$ and $\Delta z$ be numbers with
$t + \Delta t \geq 0$ and $z + \Delta z>0$. We must prove that 
$F_R(t+\Delta t, z+\Delta z) \rightarrow F_R(t,z)$ as $\Delta t, \Delta z \rightarrow 0$.
We will consider the case $\Delta t \geq 0$ and $\Delta z \geq 0$. The cases when one or both of $\Delta t$ and $\Delta z$
are negative can be understood analogously.

We have 
\begin{eqnarray*}
&~& \exists j \geq 0  ~~~ t+\Delta t<T_j\leq t+\Delta t + z + \Delta z \\
&\Rightarrow& \exists j \geq 0  ~~~t< T_j \leq t+\Delta t + z + \Delta z \\
&\Rightarrow&  \exists j \geq 0  ~~~ t<T_j \leq t+z  ~~\mbox{or}~~  t+z < T_j \leq t+z+\Delta t + \Delta z 
\end{eqnarray*}
and therefore
\begin{equation}
\label{eq:obere}
F_R(t+\Delta t, z+\Delta z)  \leq F_R(t,z) + F_R(t+z, \Delta t+ \Delta z).
\end{equation}
Similarly, 
\begin{eqnarray*}
&~& \exists j \geq 0  ~~~ t<T_j \leq t+z \\
&\Rightarrow& \exists j \geq 0  ~~~ t<T_j \leq t+\Delta t + z + \Delta z \\
&\Rightarrow&\exists j \geq 0  ~~~ t<T_j \leq t+\Delta t   ~\mbox{or}~ t+\Delta t<T_j \leq t+\Delta t + z + \Delta z 
\end{eqnarray*}
and therefore
\begin{equation}
\label{eq:untere}
F_R(t,z) \leq F_R(t,\Delta t) +  F_R(t+\Delta t, z+\Delta z). 
\end{equation}
From (\ref{eq:obere}) and (\ref{eq:untere}), 
$$
\left| F_R(t+\Delta t, z+\Delta z) - F_R(t,z) \right| \leq \max \left( F_R(t+z,\Delta t + \Delta z), F_R(t,\Delta t) \right).
$$
It therefore suffices to prove now that 
for any $a \geq 0$,
\begin{equation}
\label{eq:this_now_suffices}
\lim_{\Delta t \rightarrow 0^+} ~P \left( \exists j \geq 0~~~ a<T_j \leq a+\Delta t \right) = 0.
\end{equation}
To prove (\ref{eq:this_now_suffices}), we write, for any $n \geq 1$,
\begin{equation}
\label{eq:noch_eine_schranke}
P \left( \exists j \geq 0 ~~~ a<T_j \leq a+\Delta t \right)  \leq  \sum_{j=1}^n P \left(a< T_j \leq a+\Delta t  \right) +
P \left( T_{n+1} \leq a+\Delta t \right).
\end{equation}
Because $F$ is continuous, the right-hand side of (\ref{eq:noch_eine_schranke}) converges to 
$P(T_{n+1} \leq a)$ as $\Delta t \rightarrow 0^+$. This implies 
$$
\limsup_{\Delta t \rightarrow 0^+} ~P \left( \exists j \geq 0 ~~~a<T_j\leq a+\Delta t \right) \leq P
 \left( T_{n+1} \leq a \right).
$$
The assertion now follows because  $P(T_{n+1} \leq a) \rightarrow 0$ as $n \rightarrow \infty$.
\end{proof}

The importance of $F_0$ lies in the following well-known result from renewal theory \cite[Chapter XI, Section 4]{Feller_II}.
\begin{theorem}
In the equilibrium case, 
$$
F_R(t,z) = F_0(z) ~~ \mbox{for all $t \geq 0$, $z>0$.}
$$
In the transitional case,
$$
\lim_{t \rightarrow \infty} F_R(t,z)  = F_0(z) ~~~ \mbox{for all $z>0$.}
\label{thm:equilibrium}
$$
\end{theorem}

\pagebreak

We will use the following elementary fact. 

\begin{lemma}
\label{lemma:finite_v_F} 
If $\sigma^2 < \infty$, then $\d{ 1-F(t) = o \left( 1/t^2 \right)}$.
\end{lemma}

\begin{proof}
Let $\tau$ denote a random variable with distribution function $F$. 
Let $\Omega$ denote the underlying probability space. For $t>0$, let
$I_t(\omega) =1$ if $\tau(\omega ) >  t$, and $=0$ otherwise. Then
$$
t^2 (1-F(t)) = \int_{\tau(\omega) >t} t^2 ~\! d \omega \leq \int_{\tau(\omega) > t} \tau(\omega)^2 d \omega = 
\int_\Omega \tau(\omega)^2 I_t(\omega) d \omega.
$$
The integrand converges to zero as $t \rightarrow \infty$ for any fixed $\omega$, and is bounded by the integrable function 
$\tau(\omega)^2$. The Lebesgue dominated convergence theorem implies the assertion.
\end{proof}

Some simple properties of $F_0$ that will be of use to us later on are collected in the following lemma.

\begin{samepage}
\begin{lemma} 
\label{lemma:F_0_properties}
\begin{enumerate}
\item[(a)]
If $F$ is continuous, then $F_0$ is continuously differentiable, with
density
\begin{equation*}
\rho_0(t) =  \frac{1-F(t)}{\taubar}, ~~~ t>0.
\end{equation*}
\item[(b)]
 If $\sigma^2<\infty$ then $1-F_0(t)=o(1/t)$ as
$t \rightarrow \infty$.

\item[(c)] The expectation of $F_0$ is $ \sigma^2 /(2\taubar)$, regardless of whether $\sigma^2$ is finite
or infinite.

\end{enumerate}
\end{lemma}
\end{samepage}
\begin{proof}
(a) follows from the fundamental theorem of calculus. 

(b) Since
$
 1-F_0(t) = \frac{1}{\taubar} 
\int_t^\infty (1-F(s)) ds, 
$
this follows from Lemma \ref{lemma:finite_v_F}.

\vskip 5pt
(c) The expectation of $F_0$ is 
$$
\int_0^\infty (1-F_0(s)) ds = \left[ s(1-F_0(s)) \right]_0^\infty + \int_0^\infty s \frac{1-F(s)}{\taubar} ds
=
\lim_{s \rightarrow \infty} s (1-F_0(s)) + \frac{\sigma^2}{2 \taubar}.
$$

Using part (b), we see that this equals
$\sigma^2/(2 \taubar)$.
\end{proof}

Now define the random, piecewise constant velocity, $V$,
$$
V(t) = \left\{ \begin{array}{cl} V_0   & \mbox{if $0  \leq t  < T_0$}, \\
V_i  & \mbox{if $T_{i-1} \leq t  < T_i$}. 
\end{array}
\right.
$$
Consider a random walker in $\R^d$ starting at $X(0)=0$ and moving with
velocity $V(t)$ at time $t$. The position of the random walker at time $t$ is
$$
X(t) = 
\int_0^t V(s) ds.
$$
The {\em mean square displacement (MSD)} is the quantity
$$
M(t) = E[| X(t)| ^2].
$$
As mentioned in the introduction, we use subscripts  to distinguish 
the equilibrium and transitional cases, writing
$M_{eq}$ and $M_{tr}$ and referring to  $M_{eq}$ as the  {\em equilibrium  MSD}, and to $M_{tr}$ as the {\em transitional  MSD.}

\pagebreak

\section{Integral representations and properties of the MSD}
\label{sec:computing_M}

\subsection{The MSD as a double integral}
The MSD can be expressed as an integral as follows \cite{dechant2014scaling}. For all $t\geq0$, 
\begin{eqnarray*}
M(t)  &=& E \left[ \int_0^t V(r) dr \cdot \int_0^t V(s) ds \right]  \\
&=& 
E \left[ \int_0^t \int_0^t V(r) \cdot V(s) ~\! dr ~\! ds \right]  \\ 
&=& 
2  E \left[ \int_0^t \int_0^s V(r) \cdot V(s) ~\! dr ~\! ds \right].
\end{eqnarray*}
Recalling that expectations are integrals over the sample space and using
Fubini's theorem, we conclude that
\begin{equation}
\label{eq:central_M_of_t_formula}
M(t) = 2  \int_0^t \int_0^s E \left[ V(r) \cdot V(s) \right]  dr ~\! ds. 
\end{equation}
 The auto-correlation $E[V(r) \cdot V(s)]$ that appears in (\ref{eq:central_M_of_t_formula}) 
is related to the residual life distribution $F_R$ as follows. Assume that $0\leq r<s$. If there is no $j$ with $T_j \in (r,s]$,
an event of probability $1-F_R(r,s-r)$,  then $V(r)=V(s)$ and so $E[V(r) \cdot V(s)]=v^2$. If there is a $j$ with 
$T_j \in (r,s]$, then $V(r)$ and $V(s)$ are independent and $E[V(r) \cdot V(s)]=0$.
Hence, 
\begin{equation}
\label{eq:auto_corr_from_F_R}
E[V(r)  \cdot V(s)] =  v^2 \left( 1-F_R(r,s-r) \right) 
\end{equation}
and, inserting (\ref{eq:auto_corr_from_F_R}) into (\ref{eq:central_M_of_t_formula}), 
\begin{equation}
M(t)=2 v^2 \int_0^t \int_0^s \left( 1-F_R(r,s-r) \right)  dr ~\! ds
\label{eq:generalM}
\end{equation}
for all $t\geq0$.

\subsection{Properties of the MSD}
 
We record the following general consequences of eq.\ (\ref{eq:generalM}).

\begin{proposition}
\label{prop:props}
In both the equilibrium and transitional cases, 
\begin{enumerate}
\item[(a)]
$
M(t) \sim v^2t^2
$ as $t \rightarrow 0$, 
\item[(b)] $M$ is continuously differentiable,
\item[(c)] $M^\prime(t)>0$ for all $t>0$, 
\item[(d)] $M$ grows at least linearly, i.e., 
$
{\d \liminf_{t \rightarrow \infty} M'(t)>0}, and
$
\item[(e)] for any fixed $s>0$, 
${\d 
\lim_{t \rightarrow \infty} \frac{M(t-s)}{M(t)} = 1.}
$
\end{enumerate}
\end{proposition}
\begin{proof}
(a) means
$$
\frac{M(t)}{t^2} = \frac{2v^2}{t^2} \int_0^t \int_0^s (1-F_R(r,s-r)) ~\! dr ~\! ds \rightarrow v^2, 
$$
and this holds because 
$$
\max \left\{ F_R(r,s-r):0 \leq r < s \leq t \right\} \leq F(t)  \rightarrow 0 
$$
as $t \rightarrow 0$.
(b) follows immediately from  (\ref{eq:generalM})
because $F_R$ is a continuous function (Lemma \ref{lemma:F_R_continuous}). In fact, 
$$
M'(t) = 2 v^2 \int_0^t (1-F_R(r,t-r)) ~\! dr.
$$
Note that  $1-F_R(r,t-r)$ is a non-negative continuous function of $r$ and $t$ which tends to 1 as $r \rightarrow  t^-$. Thus,
$M'(t)>0$ for all $t>0$, establishing (c). To prove (d), 
we will prove the
existence of a positive lower bound on $M'$, valid for sufficiently large $t$.
Let $a>0$ and $t \geq a$. Then 
\begin{eqnarray*}
M'(t) &=& 2v^2 \int_0^t (1- F_R(r,t-r)) ~dr \\  &\geq& 2v^2 \int_{t-a}^t \left( 1 - F_R(r,t-r) \right)  dr  \\
\\ &\geq &
2v^2 \int_{t-a}^t  \left( 1- F_R(t-a,a) \right) dr  
 \\
&=& 
2v^2a \left(1 - F_R(t-a,a) \right) \\
&\rightarrow& 
 2v^2a (1-F_0(a))
\end{eqnarray*}
as $t \rightarrow \infty$. 
For sufficiently small $a>0$, we have $2v^2a (1-F_0(a))>0$. Choose $a$ so that this holds. Then for
sufficiently large $t$, 
$$
M'(t) \geq  v^2 a (1-F_0(a)) > 0.
$$
This implies (d). To prove (e), first note
$M(t) - M(t-s)>0$ because of (c).  Furthermore, 
\begin{eqnarray*}
(M(t) - M(t-s))^2 &=& \left(E[|X(t)|^2 - |X(t-s)|^2]\right)^2  \\
&=& \left(E[(|X(t)|-|X(t-s)|) (|X(t)|+|X(t-s)|) ] \right)^2  \\
&\leq&  {E[\left|X(t)-X(t-s)\right|^2]} ~ {E [\left|X(t)+X(t-s)\right|^2 ]}~~\\
&\leq& E\left[\left(\int_{t-s}^t|V(u)|du\right)^2\right] {2 \left(E[\left|X(t)\right|^2] +  E[\left|X(t-s)\right|^2]\right)} \\
&\leq& 2E\left[s\int_{t-s}^t|V(u)|^2du\right] {\left(E[\left|X(t)\right|^2] +  E[\left|X(t)\right|^2]\right)} \\
&=& 2s\int_{t-s}^tE\left[|V(u)|^2\right]du {\left(E[\left|X(t)\right|^2] +  E[\left|X(t)\right|^2]\right)} \\
&=& 4v^2s^2{M(t)}.
\end{eqnarray*}
Hence, $M(t)-M(t-s)\leq 2vs\sqrt{M(t)}$, and, since $M(t) \rightarrow \infty$ by (d), (e) follows.
\end{proof}

\subsection{Computing $\bm{M_{eq}}$ from $\bm{F}$}
\label{subsec:equi}
As we have seen in Theorem \ref{thm:equilibrium}, in the equilibrium case, $F_R(r,s-r)=F_0(s-r)$ for all $r$. Hence,  eq.\ (\ref{eq:generalM}) becomes
\begin{eqnarray}
\label{eq:M_{eq}_double}
M_{eq}(t) 
& = & 2v^2 \int_0^t \int_0^s (1-F_0(s-r)) ~\! dr ds .
\end{eqnarray}
We now establish two consequences of this integral expression which will allow us both to construct illustrative examples and to establish the asymptotic behavior of the MSD.
First, differentiating (\ref{eq:M_{eq}_double}) and introducing the change of variables $t-r=u$, we find 
\begin{equation}
\label{eq:M_{eq}_prime}
M_{eq}'(t) = 2v^2 \int_0^t (1-F_0(u)) du.
\end{equation}
Equation (\ref{eq:M_{eq}_prime}) implies, by anti-differentiation, a representation of $M_{eq}$ as a single integral:
\begin{equation}
\label{eq:M_{eq}_single}
M_{eq}(t) = 2 v^2 \int_0^t (1-F_0(u)) (t-u) du.
\end{equation}
Integrating by parts in (\ref{eq:M_{eq}_single}), using (\ref{eq:def_F_equi}), we 
obtain an integral expression for $M_{eq}$ directly in terms of $F$ rather than $F_0$:
\begin{proposition} For all $t\geq0$, the equilibrium MSD can be expressed as 
\label{prop:M_eq_from_F}
\begin{eqnarray}
M_{eq}(t) &=& v^2 \left( t^2 - \frac{1}{\taubar} \int_0^t (1-F(u)) (t-u)^2 du \right) 
\nonumber
\label{eq:M_eq_from_F1}
 \\
&=& \frac{v^2}{\taubar} \left( t^2 \int_t^\infty (1-F(u)) du  + \int_0^t (1-F(u)) (2tu-u^2) du \right).
\label{eq:M_eq_from_F2}
\end{eqnarray}
\end{proposition}

Proposition \ref{prop:M_eq_from_F} will prove useful in deriving asymptotic results. We now derive 
an alternative characterization of $M_{eq}$, which will prove useful in constructing examples. 
Differentiating eq.\ (\ref{eq:M_{eq}_prime}), we obtain
\begin{equation}
\label{eq:Meqpp}
M_{eq}''(t) = 2v^2 (1-F_0(t)).
\end{equation}
Differentiating a third time, we arrive at the following proposition.

\begin{proposition}
\label{propode}
$M_{eq}$ is three times continuously differentiable and for all $t\geq0$,
\begin{equation}
\label{eq:MSD_3}
M_{eq}'''(t) = \frac{2v^2}{\taubar} (F(t)-1)
\end{equation}
with boundary conditions,
\begin{equation}
\label{eq:bc_for_M_{eq}}
 M_{eq}(0) =0, ~~~ M_{eq}'(0) =0, ~~~ M_{eq}''(0) = 2 v^2, ~~~  M_{eq}''(\infty) = 0.
\end{equation}
\end{proposition}

\noindent
It may appear at first sight that $M_{eq}$ is over-determined by the four 
conditions in (\ref{eq:bc_for_M_{eq}}). However, note that (\ref{eq:MSD_3}) implies
that 
$
M_{eq}''(0) - M_{eq}''(\infty) = 2 v^2.
$
Therefore, $M_{eq}''(0) = 2 v^2$ and $M_{eq}''(\infty) = 0$ are equivalent conditions.

\section{A family of power laws with closed-form $\bm{M_{eq}}$}
\label{sec:powerlaws}
A standard assumption in the literature is that $F$ is, asymptotically, a power law. This enables derivation of asymptotic behavior of L\'evy walks, including the MSD, via Laplace transforms and the Tauberian theorem \cite{klafter_sokolov}. 

However, in the equilibrium case, for a variety of power laws, one can get a more precise result --- the exact MSD for all times. For example,  
for $\alpha>1$ and $t_0>0$, let
\begin{equation}
\label{eq:F_origin}
F(t) = 1 - \left( 1 + \frac{t}{t_0} \right)^{-\alpha}.
\end{equation}
The mean of this distribution is 
\begin{equation}
\label{eq:first_moment}
\taubar = \int_0^\infty \left( 1 + \frac{t}{t_0} \right)^{-\alpha}dt= \frac{ t_0}{\alpha-1}. 
\end{equation}
For $\alpha\leq2$, the second moment of $F$ is infinite, while for $\alpha>2$, it is finite.
Using (\ref{eq:first_moment}), 
we see that in the limit as $\alpha \rightarrow \infty$, with $\taubar$ fixed, we obtain the exponential distribution
\begin{equation}
\label{eq:alpha_is_infinity}
F(t) = 1 -  e^{-t/\taubar}. 
\end{equation}
We also note that 
$$
F_{0}(t) = \frac{1}{\taubar} \int_0^t \left( 1 -F(s) \right) ds = 1 - \left( 1+\frac{t}{t_0} \right)^{1-\alpha}.
$$

Using either Proposition \ref{prop:M_eq_from_F} or Proposition \ref{propode}, we can  compute $M_{eq}(t)$ exactly for the distributions  in (\ref{eq:F_origin}) and (\ref{eq:alpha_is_infinity}). For example, consider the case
$\alpha=3$. Using eqs.\ (\ref{eq:MSD_3}) and (\ref{eq:first_moment}), we obtain
\[
M_{eq}^{'''}(t)=\frac{4v^2}{t_0}\left(1+\frac{t}{t_0}\right)^{-3}.
\]
Integrating, using the boundary conditions (\ref{eq:bc_for_M_{eq}}), we find
\[
M_{eq}^{''}(t)=-2v^2\left(1+\frac{t}{t_0}\right)^{-2},
\]
\[
M_{eq}^{'}(t)=2v^2t_0-2v^2t_0\left(1+\frac{t}{t_0}\right)^{-1},
\]
and
\[
M_{eq}(t)=2v^2t_0^2\left(\frac{t}{t_0}-\log\left(1+\frac{t}{t_0}\right)\right).
\]

\noindent
This proves eq.\ (\ref{alpha3}) below. The other parts of Proposition \ref{theorem:canonical} are proved by similar straightforward calculations.

\begin{proposition}
\label{theorem:canonical}
For the distribution given by (\ref{eq:F_origin}) if $\alpha<\infty$, and by (\ref{eq:alpha_is_infinity}) if
$\alpha = \infty$, we have, for all $t\geq0$,

\begin{numcases} {M_{eq} (t)=}
\!\!\frac{2 (v t_0)^2 }{(2-\alpha)(3-\alpha)}  \left(\left( 1+\frac{t}{t_0} \right)^{3-\alpha}- (3-\alpha)\frac{t}{t_0} -1   \right) & \kern-1em \kern-.2em if   $\alpha>1, \alpha\!\neq\!2,3$
\label{generic}
\\ 
\label{alpha2}
\!\! 2 (v t_0)^2 \left(  \left( 1 + \frac{t}{t_0} \right) \log \left( 1 + \frac{t}{t_0} \right) - \frac{t}{t_0} \right) & \kern-1em\kern-.2em \mbox{if }  $\alpha=2$
\\
\!\! 2 (v t_0)^2 \left( \frac{t}{t_0} - \log \left( 1 + \frac{t}{t_0} \right) \right) & \kern-1em\kern-.2em \mbox{if }  $\alpha=3$
\label{alpha3}
\\
\!\! 2 (v \taubar)^2 \left( \frac{t}{\taubar} - 1 + e^{-t/\taubar} \right) & \kern-1em\kern-.2em \mbox{if }  $\alpha=\infty$
\label{alphainf}
\end{numcases}
\end{proposition}

We stress again that these are {\em exact} formulas valid for all times, not just {\em asymptotic} ones.
Note that (\ref{generic}) tends to (\ref{alpha2}) as $\alpha \rightarrow 2$, and to (\ref{alpha3}) as $\alpha \rightarrow 3$.
Using (\ref{eq:first_moment}), 
it is also straightforward to see that in the limit as $\alpha \rightarrow \infty$ with $\taubar$ fixed, (\ref{generic}) tends to 
(\ref{alphainf}).

When the particles move at a constant speed $v$, there is a sharp front consisting of particles which have experienced no collisions. Then at time $t$, the probability of having experienced no collisions is $1-F_0(t)$ and the contribution to $M_{eq}$ of this front is $v^2t^2(1-F_0(t))$. For the above power law with $1<\alpha<2$, this gives a contribution of 
$$
M_{eq}^f(t)= v^2 t^2  \left( 1 + \frac{t}{t_0} \right)^{1-\alpha}.
$$
As $t\rightarrow\infty$, the front contributes an asymptotically constant proportion to the total MSD:
$$
\frac{M_{eq}^f(t)}{M_{eq}(t)}\rightarrow\frac{(2-\alpha)(3-\alpha)}{2}.
$$
For $\alpha\geq2$, the  contribution of the front to $M_{eq}(t)$ becomes negligible as $t \rightarrow \infty$.

\section{Leading-order asymptotics of $\bm{M_{eq}}$}
If $F$ has finite variance, then
$$
\lim_{t \rightarrow \infty} M_{eq}'(t)=2v^2\lim_{t \rightarrow \infty}\int_0^t (1-F_0(s)) ds=\frac{\sigma^2 v^2}{\taubar},
$$
by eq.\ (\ref{eq:M_{eq}_prime}) together with 
Lemma \ref{lemma:F_0_properties}, part (c).
Hence,
\begin{equation}
\label{eq:finite_variance}
M_{eq}(t) \sim \sigma^2 v^2 \frac{t}{\taubar} ~~~ \mbox{as $t \rightarrow \infty$}.
\end{equation}

The more interesting case is that of infinite variance, to which we now turn. In several
places we will make use of the following lemma. 

\begin{lemma}
\label{lemma:integrate_asymptotics}
Let $f$ and $g$ be non-negative, locally integrable functions of $t \in [0,\infty)$ with 
$f(t) > 0 $ and $g(t)>0$ for sufficiently large $t$,
$$
f(t) \sim g(t) ~~~ \mbox{as $t \rightarrow \infty$}, 
$$
and 
$$
\int_0^\infty f(t) dt = \int_0^\infty g(t) dt = \infty.
$$
Then 
$$
\int_0^t f(s) ds \sim \int_0^t g(s) ds ~~~ \mbox{as $t \rightarrow \infty$.}
$$
\end{lemma}

\begin{proof} Let $\epsilon >0$. Choose $A>0$ so that for $t \geq A$, 
$
 f(t) \leq (1+\epsilon) g(t). 
$
Then for $t \geq A$, 
\begin{eqnarray*}
\int_0^t f(s) ds = \int_0^A f(s) ds + \int_A^t f(s) ds &\leq&
\int_0^A f(s) ds + (1+\epsilon) \int_A^t g(s) ds \\ &\leq&
\int_0^A f(s) ds + (1+\epsilon) \int_0^t g(s) ds. 
\end{eqnarray*}
Therefore 
$$
\frac{\int_0^t f(s) ds}{\int_0^t g(s) ds} \leq \frac{\int_0^A f(s)ds}{\int_0^t g(s)ds} +1 + \epsilon.
$$
Letting $t \rightarrow \infty$, we find 
$$
\limsup_{t \rightarrow \infty} \frac{\int_0^t f(s) ds}{\int_0^t g(s) ds}  \leq 1 + \epsilon.
$$
By a similar argument, 
$$
\liminf_{t \rightarrow \infty} \frac{\int_0^t f(s) ds}{\int_0^t g(s) ds}  \geq 1 - \epsilon.
$$
Since these conclusions hold for any $\epsilon>0$, the assertion follows.
\end{proof}

\pagebreak

The following theorem shows that in the infinite variance case, the leading-order long-time asymptotic
behavior 
of $\taubar M_{eq}$ is determined by the leading-order long-time asymptotic
behavior of $1-F$.

\begin{theorem}
\label{theorem:F_F_tilde}
Let $F$ and $\tilde{F}$ be infinite-variance distributions with
finite means $\taubar$ and $\tilde{\taubar}$ and equilibrium MSDs $M_{eq}$ and $\tilde{M}_{eq}$, respectively.
Assume that
$$
1-F(t) ~\sim ~1-\tilde{F}(t) ~~~ \mbox{as $t \rightarrow \infty$}.
$$
Then
$$
\taubar M_{eq}(t) \sim \tilde{\taubar} \tilde{M}_{eq}(t) ~~~ \mbox{as $t \rightarrow \infty$}.
$$
\end{theorem}

\begin{proof} 
$
\taubar M_{eq}'''= 2v^2 (F-1)$, $M_{eq}''(\infty)=0$ and $\tilde{\taubar} \tilde{M}_{eq}'''= 2v^2 (\tilde{F}-1)$, $\tilde{M}_{eq}''(\infty)=0$ imply 
$$
\taubar M_{eq}''(t) = \int_t^\infty 2 v^2 (1-F(s) )~ ds ~~~\sim~~~ \tilde{\taubar} \tilde{M}_{eq}''(t) 
= \int_t^\infty 2 v^2 (1-\tilde{F}(s) )~ ds
$$
as $t \rightarrow \infty$. From Lemma \ref{lemma:F_0_properties}, part (c), and eq.\!\! (\ref{eq:M_{eq}_prime}), we know that  $M_{eq}'(t)$ and $\tilde{M}_{eq}'(t)$ tend to $\infty$ as $t \rightarrow \infty$. This implies, using Lemma \ref{lemma:integrate_asymptotics},  that
$$
\taubar M_{eq}'(t) =  \int_0^t \taubar M_{eq}''(s) ds ~ \sim ~ 
\tilde{\taubar} \tilde{M}_{eq}'(t) =  \int_0^t \tilde{\taubar} \tilde{M}_{eq}''(s) ds.
$$
Because $M_{eq}(t)$ and $\tilde{M}_{eq}(t)$ tend to $\infty$, this in turn implies, again by Lemma \ref{lemma:integrate_asymptotics}, 
$$
\taubar M_{eq}(t) =  \int_0^t \taubar M_{eq}'(s) ds ~ \sim ~ 
\tilde{\taubar} \tilde{M}_{eq}(t) =  \int_0^t \tilde{\taubar} \tilde{M}_{eq}'(s) ds, 
$$
which is the statement of the theorem.
\end{proof}

Proposition \ref{theorem:canonical} and Theorem \ref{theorem:F_F_tilde} together imply the well-known formulas \cite{barkai1997levy,shlesinger1985comment} for the leading-order asymptotic behavior of $M_{eq}$ for $\alpha \in (1,2]$, summarized in the following corollary. We add to the statement of this corollary the result for the finite variance case $\alpha>2$ (see eq.\ (\ref{eq:finite_variance})).

\begin{corollary} 
\label{corollary:M_eq_asymp} 
Suppose that
$$
1-F(t) \sim \left( \frac{t}{t_0} \right)^{-\alpha}
$$
for some $t_0>0$ and $\alpha>1$. 
Then as $t \rightarrow \infty$,
\begin{numcases}
{M_{eq}(t)\sim}
\frac{t_0}{\taubar}~\!
\frac{2 (v t_0)^2}{(\alpha-1) (2-\alpha) (3-\alpha)} \left( \frac{t}{t_0} \right)^{3-\alpha}
& \mbox{if $1<\alpha<2$},
\label{eq:alpha2}
\\
\frac{t_0}{\taubar}  2 (v t_0)^2 \frac{t}{t_0} \log \frac{t}{t_0} & \mbox{if $\alpha = 2$}, 
\label{5.3}
\\
\frac{t_0}{\taubar} (v \sigma)^2 \frac{t}{t_0}
& \mbox{if $\alpha > 2$}.
\label{5.4}
\end{numcases}

\end{corollary}


\pagebreak
\section{Accuracy of leading-order asymptotics: Examples.} 
In this section we explore the accuracy of 
leading-order asymptotic descriptions of $M_{eq}$ such as those given in Corollary \ref{corollary:M_eq_asymp}. 
 
\subsection{Canonical power laws} \label{subsec:power} For the power laws
$$
F(t) = 1 -  \left( 1 + \frac{t}{t_0} \right)^{-\alpha}
$$
discussed in Section \ref{sec:powerlaws}, the accuracy of the leading-order asymptotics can be ascertained from 
eqs.\ (\ref{generic})--(\ref{alpha3}). The discrepancy between the exact $M_{eq}$ and the leading-order
asymptotics given in Corollary \ref{corollary:M_eq_asymp} is $O(t/t_0)$ for $\alpha \leq 2$, and $O((t/t_0)^{3-\alpha})$
for $2 < \alpha < 3$.
Therefore the leading-order asymptotic
approximation for $M_{eq}(t)$
 is inaccurate when $\alpha \approx 2$ unless $t/t_0$ is very large.
 
 Figure \ref{fig:MC_F_ALPHA_EQUILIBRIUM} illustrates this conclusion, using $t_0=1$, $v=1$, and six different values of $\alpha$. The figure shows, for $0 \leq t \leq 1000$,  the exact $M_{eq}(t)$  (solid),  the leading-order approximations
(dotted), 
and Monte Carlo estimates of $M_{eq}$ (circles), based on 500,000 simulated random walkers  for each value of $\alpha$.

\begin{figure}[ht!]
\begin{center}
\includegraphics[scale=0.6]{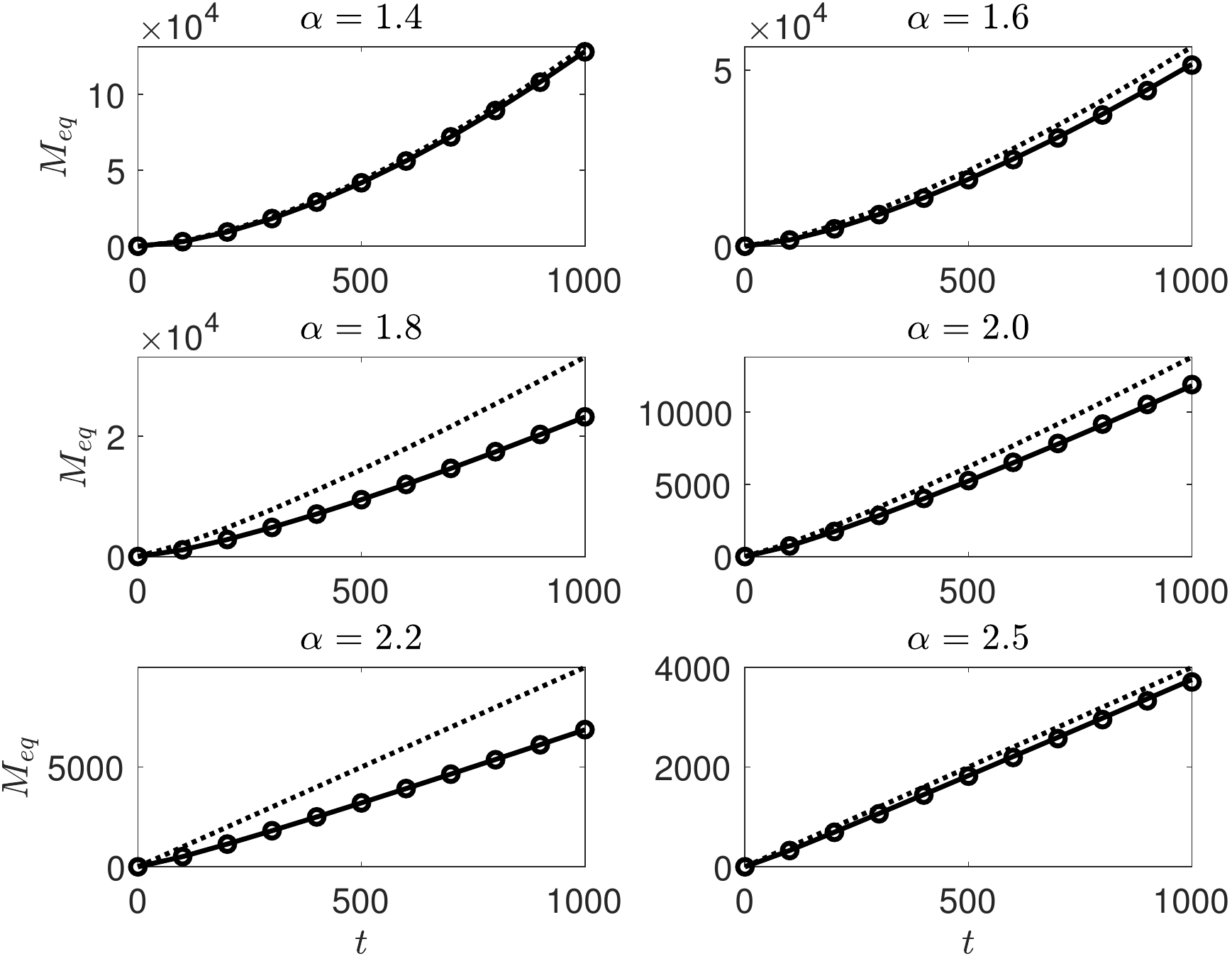}
\end{center}
\caption{Exact $M_{eq}$ (solid), leading-order approximation
as given in Corollary \ref{corollary:M_eq_asymp}  (dotted), and approximations of $M_{eq}$ 
obtained by simulations of 500,000 walkers (circles) for six different values of $\alpha$.}
\label{fig:MC_F_ALPHA_EQUILIBRIUM}
\end{figure}

For $\alpha=2$, the relative difference between $M_{eq}$ and 
its leading-order approximation is approximately $1/\log(t/t_0)$, and for this to be smaller than 0.05, 
corresponding to 5\% accuracy, we need $t/t_0>e^{20} \approx 5 \times 10^8$. Zarfaty et al.\!\! \cite{Zarfaty:2018tz} relate slow convergence of $M_{eq}$ to slow convergence in the Central Limit Theorem (CLT) for $\alpha=2$ \cite{Keller2001rate}, a special case of slow convergence for any slowly varying scaling in the generalized CLT \cite{borgers2018slow}. 
A related, deterministic, problem is that of infinite-horizon billiards, in which the asymptotic MSD also contains a logarithmic factor. For this 
case, Cristadoro et al.\@ derived the first two terms in the asymptotic expansion 
of the MSD \cite{Cristadoro_et_al_2014} and demonstrated numerically that the leading-order term is a poor approximation of the actual MSD for reasonable lengths of time \cite{cristadoro2014measuring}. In \cite{cristadoro2014transport}, the same authors derived an evolution equation for the MSD for random walks on a lattice with exponentially distributed waiting times between flight segments, a problem that can be seen as an abstraction of the 
infinite-horizon Lorentz gas when $\alpha=2$.

\subsection{A general approach to constructing examples}

Instead of beginning with $F$ and computing $M_{eq}$, we can also, for the purposes of constructing
interesting examples, start with $M_{eq}$ and compute $F$ from it, based on the following proposition.

\begin{proposition}
\label{theorem:gamma_conditions}
Let $\mathcal{M} \in C^3([0,\infty))$. 
There exist an expected  speed $v>0$ and a continuous probability distribution function 
$F$ on $(0,\infty)$  with finite mean $\taubar$ such that $M_{eq}=\mathcal{M}$ is the equilibrium MSD
associated with $v$ and $F$ if and only if
\begin{enumerate}
\item $\mathcal{M}(0) = \mathcal{M}'(0)=0$, 
\item $\mathcal{M}''(0)>0$ and $\mathcal{M}''(\infty)=0$, 
\item $\mathcal{M}'''$ is an increasing function with $\mathcal{M}'''(0)<0$ and $\mathcal{M}'''(\infty)=0$.
\end{enumerate}
In that case, 
\begin{equation*}
\label{eq:defF}
v^2 = \frac{\mathcal{M}''(0)}{2}, ~~~
F(t) = 1 - \frac{\mathcal{M}'''(t)}{\mathcal{M}'''(0)}, ~~~
\taubar =  \frac{\mathcal{M}''(0)}{ \left| \mathcal{M}{~\!'''}(0) \right|}.
\end{equation*}
\end{proposition}

\begin{proof} 
This is a straightforward consequence of Proposition \ref{propode}.
\end{proof}

\subsection{Accurate leading-order asymptotics}
Examples in which there are accurate leading-order asymptotic descriptions of $M_{eq}$ can be constructed using 
Proposition
 \ref{theorem:gamma_conditions}. For instance, 
 $$
M_{eq} = v^2 t \log(1+t)  = v^2 t \log t + O(1)
$$
when 
$$
F(t) = 1 - \frac{3+t}{3(1+t)^3}, 
$$
and for $1 < \alpha < 2$, 
$$
M_{eq}(t) = v^2t^2(1+t)^{1-\alpha} = v^2t^{3-\alpha}+O(t^{2-\alpha}) 
$$
when 
$$
F(t) = 1-(1+t)^{-\alpha}\left(1-\frac{\alpha t}{1+t}+\frac{\alpha(\alpha+1)t^2}{6(1+t)^2}\right).
$$

\subsection{Logarithmic factors for $\bm{\alpha \in (1,2)}$}
\label{sec:logfactor}
What happens when the distribution function is asymptotically close to, but not exactly, a power law? For example, set
$$
F(t) = 1 - \frac{1+\log(1+t)}{(1+t)^{3/2}}.
$$
A closed-form equilibrium MSD can be derived for this distribution function, as well:
\begin{eqnarray*}
\nonumber
M_{eq}(t) &=& v^2 \left( \frac{8}{9} \left(t +1 \right)^{3/2} \log \left( t + 1 \right) + 
\frac{8}{27} \left( t + 1 \right)^{3/2} - \frac{4}{3} t - \frac{8}{27} \right) \\
\label{eq:compare}
&=& v^2 \frac{8}{9} t^{3/2}  \left(  \log \left( t  \right) + \frac{1}{3} \right)
+ O(t).
\end{eqnarray*} The leading-order approximation
$
\label{eq:alpha_3_2}
M_{eq}(t) \sim v^2 \frac{8}{9} t^{3/2} \log t
$
is very poor because the next-most-significant term is $O(t^{3/2})$.

\section{Super-linear, linear, and sub-linear terms}

We saw in the last section that linear correction terms can have considerable implications for the accuracy of the leading-order approximation of $M_{eq}$ when $\alpha$ is near 2. Here we show that that for $\alpha \in (1,2]$, the linear contributions to $M_{eq}$ depend on the entire distribution $F$, 
while super-linear terms, for a given $\taubar$, only depend on the asymptotic behavior of $F$.

The following theorem shows that the entire distribution $F$ is needed to compute linear contributions to $M_{eq}$: The value of  
$\taubar$ and the precise tail of $F$ (i.e., $F(t)$ for $t \geq A$, for some $A>0$) do not determine the linear 
contributions.

\begin{theorem}
\label{thm:A}
Let $F$ be continuously differentiable, with equilibrium MSD $M_{eq}$. Then for all $A>0$ such that $F(A)>0$, there exists a distribution function $\tilde{F}$, with equilibrium
MSD $\tilde{M}_{eq}$, such that the means of $F$ and $\tilde{F}$ agree, $F$ and $\tilde{F}$ agree on $[A,\infty)$, and 
$$
|M_{eq}(t)-\tilde{M}_{eq}(t)|\geq O(t)
$$
as $t \rightarrow \infty$.
\end{theorem}

\begin{proof}
By eq.\ (\ref{eq:M_eq_from_F2}), if $F$ and $\tilde F$ agree on $[A,\infty)$, then for all $t\geq A$,
$$
M_{eq}(t)-\tilde{M}_{eq}(t)=-\frac{v^2}{\taubar}\int_0^A( F (u)-\tilde F(u))(2tu-u^2) du.
$$
Hence, to complete the proof, it suffices to find an increasing $\tilde F$ such that 
$$
\int_0^A(F(u)-\tilde F (u))du=0 ~~~~\mbox{but}~~~~\int_0^A (F(u)-\tilde F(u))udu\neq0,
$$
which is, of course, always possible.
\end{proof} 

On the other hand, the following theorem shows that the value of $\taubar$ and knowledge of $F$ up to a ``finite variance" piece are sufficient to determine 
super-linear contributions to $M_{eq}$.
\begin{theorem}
\label{theorem:general_M_{eq}_expansion}
Let
$$
1-F(t) = G(t) + H(t), 
$$
where $G$ and $H$ are integrable functions of $t>0$ with 
\begin{equation}
\label{eq:finite_variance_sort_of}
\int_0^\infty |H(t)| t ~\! dt < \infty.
\end{equation}
Then 
\begin{equation}
\label{eq:M_{eq}_general_expansion}
M_{eq}(t) = 
\frac{v^2}{\taubar} \left( t^2\int_t^\infty G(u) du ~ + \int_0^t G(u) (2tu-u^2) ~\!du  \right) + O(t).
 \end{equation}
\end{theorem}

Assumption (\ref{eq:finite_variance_sort_of}) makes precise the 
assertion that $H$ is a ``finite variance piece"; compare  eq.\ (\ref{eq:sigma_in_terms_of_F}).

\begin{proof} 
Because of Proposition \ref{prop:M_eq_from_F} and assumption
(\ref{eq:finite_variance_sort_of}), it is enough to prove 
$$
\left| \int_t^\infty H(u) du \right|  = O \left( \frac{1}{t} \right)
~~~\mbox{and} ~~~ \left| 
\int_0^t H(u) u^2 du \right| = O(t). 
$$
Both follow from (\ref{eq:finite_variance_sort_of}): 
$$
\left| \int_t^\infty H(u) du \right|  \leq \int_t^\infty |H(u)| u \frac{1}{u}  du \leq \int_t^\infty |H(u)| u du \frac{1}{t} = o\left( \frac{1}{t} \right)
$$
and 
$$
\left| 
\int_0^t H(u) u^2 du \right|  \leq \int_0^t |H(u)| u du ~\! t \leq \int_0^\infty |H(u)| u du ~\! t = O(t).
$$
\end{proof}

In particular, this proposition implies 
that an asymptotic expansion of the slowly converging part of $1-F$
translates into an asymptotic expansion of the super-linear part of $M_{eq}$: 

\begin{corollary} 
Suppose that 
\begin{equation}
\label{eq:F_asymptotics}
1-F(t) =
\sum_{j=1}^n C_j t^{-p_j}  + C t^{-2} + O(t^{-q})
\end{equation}
as $t \rightarrow \infty$, 
where $n \geq 0$,
$$
1 < p_1<\ldots < p_n < 2,  ~~~~ q > 2,
$$
and the $C_j$ and $C$ are constants.
Then 
\begin{equation}
\label{eq:M_{eq}_asymptotics}
M_{eq}(t) = 
\frac{2v^2}{\taubar} \left[  \sum_{j=1}^n \frac{C_j}{(p_j-1) (2-p_j) (3-p_j)} t^{3-p_j} + C t \log t 
 \right] + O(t).
 \end{equation}
\end{corollary}

\begin{proof} 
We set 
$$
G(t) = 
\left\{ \begin{array}{cl} 
\sum_{j=1}^n C_j t^{-p_j}  + C t^{-2}  & \mbox{for $t \geq 1$}, \\
0 & \mbox{for $t<1$},
\end{array}
\right.
$$
and
$$
H(t) = 1-F(t) - G(t) = O(t^{-q}).
$$
Theorem \ref{theorem:general_M_{eq}_expansion} implies that
$$
M_{eq}(t) = \frac{v^2}{\taubar} \left[t^2  \int_t^\infty  \left( \sum_{j=1}^n C_j u^{-p_j}  + C u^{-2}  \right) du + 
2t\int_1^t  \left( \sum_{j=1}^n C_j u^{-p_j}  + C u^{-2}  \right) u du  
\right.
$$
$$
\left.  - \int_1^t  \left( \sum_{j=1}^n C_j u^{-p_j}  + C u^{-2}  \right) u^2 du  \right] +O(t).
$$
The assertion  follows by evaluating the integrals.
\end{proof}

What happens if $F$ doesn't follow a power law exactly?  Power law bounds on $F$ give us bounds on the MSD. For instance,
we have the following result.

\begin{corollary} 
\begin{enumerate} 
\item[(a)]
Let $\alpha \in (1,2)$, $C>0$. If 
$$
1-F(t) \leq C t^{-\alpha}
$$
for all $t >0 $, then
$$
M_{eq}(t) \leq  \frac{2v^2}{\taubar}  \frac{C}{(\alpha-1)(2-\alpha)(3-\alpha)} t^{3-\alpha}
$$
for all $t>0$.
Similarly, if $1-F(t) \geq C t^{-\alpha}$
for all $t >0 $, then the reverse inequality holds.
\item[(b)]
Let $C>0$ and $A >0$. If 
$$
1-F(t) \leq C t^{-2}
$$
for all $t \geq A$, then
$$
M_{eq}(t) \leq  \frac{v^2}{\taubar} \left[ CA+ 
\int_0^A (1-F(u)) (2tu-u^2) du
+ 2C t \log \frac{t}{A} \right]
$$
for $t \geq A$. Similarly, if $1-F(t) \geq C t^{-2}$ for all $t\geq A$, then the reverse inequality holds.

\end{enumerate}
\end{corollary}

\begin{proof}
These are immediate consequences of Proposition \ref{prop:M_eq_from_F}.
\end{proof}

\section{Asymptotic behavior of $\bm{M_{tr}}$}
\label{sec:Mn}
We next relate $M_{eq}$ and $M_{tr}$ to each other, in order to deduce asymptotic behavior of $M_{tr}$ from the asymptotic behavior of $M_{eq}$. (See also \cite{barkai1997levy}, where expressions for $M_{eq}$ and $M_{tr}$ are obtained from assumed forms of the Laplace transform of the duration density in a more general setting.)
Let $t>0$ and denote by $X_{eq}$ the position of the equilibrium process. Then
the conditional expectation of $X_{eq}(t)^2$, given that $T_0>t$, is $\left(vt\right)^2$. 
 If $T_0 \leq t$, then $X_{eq}(t)$ consists of
the initial segment of duration $T_0$, and the rest. The displacements experienced in these
two segments are not independent --- when $T_0$ is larger, the second segment is briefer. 
They are, however, uncorrelated, and therefore the variances of the displacements add. 
Writing 
as before $\rho_0 = F_0'=(1-F)/\taubar$, we have
\begin{equation}
\label{eq:M_{eq}_from_M}
M_{eq}(t) = (1-F_0(t)) (vt)^2 +  \int_0^t \rho_0(s) ((vs)^2+M_{tr}(t-s)) ds.
\end{equation}
Equation (\ref{eq:M_{eq}_from_M}) can be read as an equation representing $\rho_0 \ast M_{tr}$ 
in terms of $M_{eq}$:
$$
(\rho_0 \ast M_{tr})(t) = M_{eq}(t) - (1-F_0(t)) (vt)^2 -  \int_0^t \rho_0(s) (v s)^2 ds.
$$
By integration by parts, this can be re-written more simply: 
\begin{equation}
\label{eq:smoothed_M_tr}
(\rho_0 \ast M_{tr})(t) = M_{eq}(t) - 2v^2 \int_0^t (1-F_0(s)) ~\! s ~\! ds.
\end{equation}

\begin{lemma}
\label{lemma:rho_ast_M_{tr}_sim_M_{tr}}
$$
\rho_0 \ast M_{tr}(t) \sim M_{tr}(t) ~~~ \mbox{as $t \rightarrow \infty$.}
$$
\end{lemma}

\begin{proof} It is clear that
$$
\rho_0 \ast M_{tr}(t) = \int_0^t \rho_0(s) M_{tr}(t-s) ds \leq M_{tr}(t)
$$
for all $t$, since $M_{tr}$ is strictly increasing (Proposition \ref{prop:props}, part (c)). 
Now let $\epsilon>0$. Let $S$ be so large that
$$
\int_0^S \rho_0(s) ds \geq 1-\epsilon.
$$
Then for $t \geq S$, 
$$
\rho_0 \ast M_{tr}(t) = 
\int_0^t \rho_0(s) M_{tr}(t-s) ds \geq \int_0^S \rho_0(s) M_{tr}(t-s) ds  \geq (1-\epsilon) M_{tr}(t-S).
$$
For $t$ sufficiently large, this is 
$
\geq (1-\epsilon)^2 M_{tr}(t)
$
since $M_{tr}(t-S) \sim M_{tr}(t)$ as $t \rightarrow \infty$ (Proposition \ref{prop:props}, part (e)). 
Since the above arguments hold for all $\epsilon>0$, 
the assertion follows.
\end{proof}

\begin{theorem}
\label{theorem:finite_variance}
If $\sigma^2 < \infty$, then 
$$
M_{tr}(t) \sim M_{eq}(t) \sim \sigma^2 v^2  \frac{t}{\taubar}
$$
as $t \rightarrow \infty$.
\end{theorem}

\begin{proof}
This follows from eq.\ (\ref{eq:smoothed_M_tr}), together with Lemma \ref{lemma:rho_ast_M_{tr}_sim_M_{tr}},
Lemma \ref{lemma:F_0_properties}, part (b), and eq.\ (\ref{eq:finite_variance}).
\end{proof}

For power laws with exponent $\alpha\in(1,2)$, the smaller $\alpha-1$ (and hence the greater the first step of the equilibrium process), the more we might expect $M_{eq}$ to exceed $M_{tr}$. For fixed $\alpha$, the two MSDs are asymptotically proportional, but, indeed, the asymptotic ratio of $M_{tr}$ to $M_{eq}$ is $\alpha-1$, as seen in  
\cite{barkai1997levy,detcheverry2017generalized,Froemberg2013random}.
\begin{theorem} 
\label{theorem:algebraic_tail} Let
\begin{equation}
\label{eq:algebraic_tail}
1-F(t) \sim  \left( \frac{t}{t_0} \right)^{-\alpha}
\end{equation}
as $t \rightarrow \infty$, for some constant $t_0>0$. Then

\begin{enumerate}
\item[(a)] $
M_{tr} \sim (\alpha-1) M_{eq}
$ if $\alpha \in (1,2)$, and

\item[(b)] $
M_{tr} \sim M_{eq}
$ if $\alpha \geq 2$. 
\end{enumerate}
\end{theorem}

\begin{proof} 
(a) Let  $\alpha \in (1,2)$.
From (\ref{eq:algebraic_tail}),
$$
1- F_0(t) = \frac{1}{\taubar} \int_t^\infty (1-F(s)) ds \sim   \frac{t_0}{\taubar (\alpha-1)} \left( \frac{t}{t_0}
\right)^{-\alpha+1}.
$$
This yields, using Lemma \ref{lemma:integrate_asymptotics}, the asymptotic behavior of the second term on the right-hand side of eq.\
(\ref{eq:smoothed_M_tr}): 
$$
2 v^2 \int_0^t \left( 1 - F_0(s) \right)s ~\! ds \sim  \frac{2 t_0 (v t_0)^2}{\taubar (\alpha-1)(3-\alpha)} 
\left( \frac{t}{t_0} \right)^{3-\alpha}.
$$
From eq.\ (\ref{generic}) and Theorem \ref{theorem:F_F_tilde},
$$
M_{eq}(t) \sim  \frac{t_0/(\alpha-1)}{\taubar}~\! \frac{2 (v t_0)^2}{(2-\alpha)(3-\alpha)} \left( \frac{t}{t_0} \right)^{3-\alpha}.
$$
Therefore, by eq.\ (\ref{eq:smoothed_M_tr}) and Lemma \ref{lemma:rho_ast_M_{tr}_sim_M_{tr}}, 
$$
M_{tr} \sim   \frac{t_0}{\taubar} \frac{2 (v t_0)^2}{(\alpha-1)(2-\alpha)(3-\alpha)} \left( \frac{t}{t_0} \right)^{3-\alpha}
-  \frac{t_0}{\taubar} \frac{2 (v t_0)^2}{(\alpha-1)(3-\alpha)} 
\left( \frac{t}{t_0} \right)^{3-\alpha} \sim
 (\alpha-1) M_{eq}.
$$
(b) \ For $\alpha>2$, we have $\sigma^2 < \infty$, and therefore the assertion
follows from Theorem \ref{theorem:finite_variance}. For $\alpha=2$, 
$$
1- F_0(t) = \frac{1}{\taubar} 
\int_t^\infty (1-F(s)) ds \sim   \frac{t_0^2}{\taubar t}
$$
and therefore, using Lemma \ref{lemma:integrate_asymptotics}, 
\begin{equation}
\label{eq:extra_term}
2 v^2 \int_0^t \left( 1 - F_0(s) \right)s ~\! ds \sim 2  (v t_0)^2 \frac{t}{\taubar}.
\end{equation}
By eq.\ (\ref{alpha2}) and Theorem \ref{theorem:F_F_tilde},
$$
M_{eq}(t) \sim 2  (v t_0)^2 \frac{t}{\taubar} \log \frac{t}{t_0}.
$$
Therefore, by eq.\ (\ref{eq:smoothed_M_tr}) and Lemma \ref{lemma:rho_ast_M_{tr}_sim_M_{tr}}, 
$M_{tr} (t) \sim M_{eq}(t)$.
\end{proof}

For specific examples, (\ref{eq:smoothed_M_tr}) enables us to compute $\rho_0 \ast M_{tr}$, but not $M_{tr}$. We know from 
Lemma \ref{lemma:rho_ast_M_{tr}_sim_M_{tr}} that the
leading-order asymptotic behavior of $\rho_0 \ast M_{tr}$
is the same as that of $M_{tr}$. However, in special cases, knowledge of 
$\rho_0 \ast  M_{tr}$ yields more than just the leading-order asymptotic
behavior. We give the following example.

\begin{lemma}
\label{lemma:anti_zafarty}
Assume that
$$
1 - F(t) \sim \left( \frac{t}{t_0} \right)^{-2} ~~~ \mbox{as $t \rightarrow \infty$},
$$
for some $t_0>0$. 
Then there exists a constant $C>0$  such that for all sufficiently large $t$,
$$
0 \leq M_{tr}(t) - \rho_0 \ast M_{tr}(t) \leq C \sqrt{t(\log t)^3}.
$$
\end{lemma}

\begin{proof}
We noted in the proof of Lemma \ref{lemma:rho_ast_M_{tr}_sim_M_{tr}} that
$\rho_0 \ast M_{tr}(t)$ is a lower bound on $M_{tr}(t)$, so 
$0 \leq M_{tr}(t) - \rho_0 \ast M_{tr}(t)$.
We will now find an upper bound on 
$M_{tr}(t) - \rho_0 \ast M_{tr}(t)$: 
\begin{eqnarray}
\nonumber
M_{tr}(t) - \rho_0 \ast M_{tr}(t)  
&=& 
\int_0^\infty \rho_0(s) M_{tr}(t) ds - \int_0^t \rho_0(s) M_{tr}(t-s) ds \\
\label{eq:intermediate_stop}
&=& \int_0^t \rho_0(s) (M_{tr}(t) - M_{tr}(t-s)) ds + \int_t^\infty \rho_0(s) ds ~M_{tr}(t). ~~~~~~~
\end{eqnarray}
In the proof of Proposition \ref{prop:props}, we saw:
$$
M_{tr}(t) - M_{tr}(t-s)<
2 vs ~\! \sqrt{M_{tr}(t)}.
$$
Using this in (\ref{eq:intermediate_stop}), we find the upper bound
\begin{equation}
\label{eq:stop_again}
 2 v \sqrt{M_{tr}(t)} \int_0^t s \rho_0(s) ds  +  \int_t^\infty \rho_0(s) ds ~M_{tr}(t) 
\end{equation}
By hypothesis,
\begin{equation}
\label{eq:zero}
\rho_0(t) \sim \frac{1}{\taubar} \left( \frac{t}{t_0} \right)^{-2}
\end{equation}
as $t\rightarrow\infty$.
This implies 
\begin{equation}
\label{eq:one}
\int_0^t s \rho_0(s) ds \sim \frac{t_0^2}{\taubar} \log t.
\end{equation}
We know from Theorem \ref{theorem:algebraic_tail} that
\begin{equation}
\label{eq:two}
M_{tr}(t) \sim C_1 t \log t
\end{equation}
for some positive constant $C_1$.  
Furthermore, from (\ref{eq:zero}), 
\begin{equation}
\label{eq:three}
\int_t^\infty \rho_0(s) ds \sim \frac{t_0^2}{\taubar} ~\! t^{-1}.
\end{equation}
Using (\ref{eq:one})--(\ref{eq:three}) in 
(\ref{eq:stop_again}), we obtain the assertion. 
\end{proof}

As was the case for equilibrium L\'evy walks, even exact knowledge of the tails and means of step durations doesn't suffice to determine the linear terms of the asymptotics of the transitional MSDs:
\begin{theorem} 
Let a distribution $F$, with mean $\taubar$, satisfy
$$
1-F(t) \sim \left( \frac{t}{t_0} \right)^{-2}
$$
as $t\rightarrow\infty$, for some $t_0>0$, and let $A>0$ be given. Then there exists a distribution $\tilde F$ with the same mean, $\taubar$, and with
$$
F(t) = \tilde{F}(t) ~~~\mbox{for all $t \geq A$},
$$
such that the corresponding mean square
displacements $M_{tr}$ and $\tilde{M_{tr}}$ differ by at least $O(t)$.
\end{theorem}

\begin{proof} 
Let $\tilde F$ be as in Theorem \ref{thm:A}.
By Lemma \ref{lemma:anti_zafarty} and
eqs.\ (\ref{eq:smoothed_M_tr}) and (\ref{eq:extra_term}), 
$$
M_{tr}(t)  =  M_{eq}(t) - 2 (v t_0)^2 \frac{t}{\taubar} + o(t), 
$$
and similarly
$$
\tilde{M}_{tr}(t)  =  \tilde{M}_{eq}(t) - 2 (v t_0)^2 \frac{t}{\taubar} + o(t), 
$$
as $t \rightarrow \infty$. Since, by Theorem \ref{thm:A},  $M_{eq}(t)$ and $\tilde{M}_{eq}(t)$
differ by $\geq O(t)$, so do $M_{tr}$ and $\tilde {M}_{tr}$.
\end{proof}

For the canonical power laws of Section \ref{subsec:power}, 
Figure \ref{fig:MC_F_ALPHA_TRANSITIONAL} illustrates  our results by showing, for six different values of 
$\alpha$, the exact MSD (solid) and the leading order asymptotic approximation (dots) for the equlibrium (black) and transitional 
(dots) cases.

\begin{figure}[ht!]
\begin{center}
\includegraphics[scale=0.6]{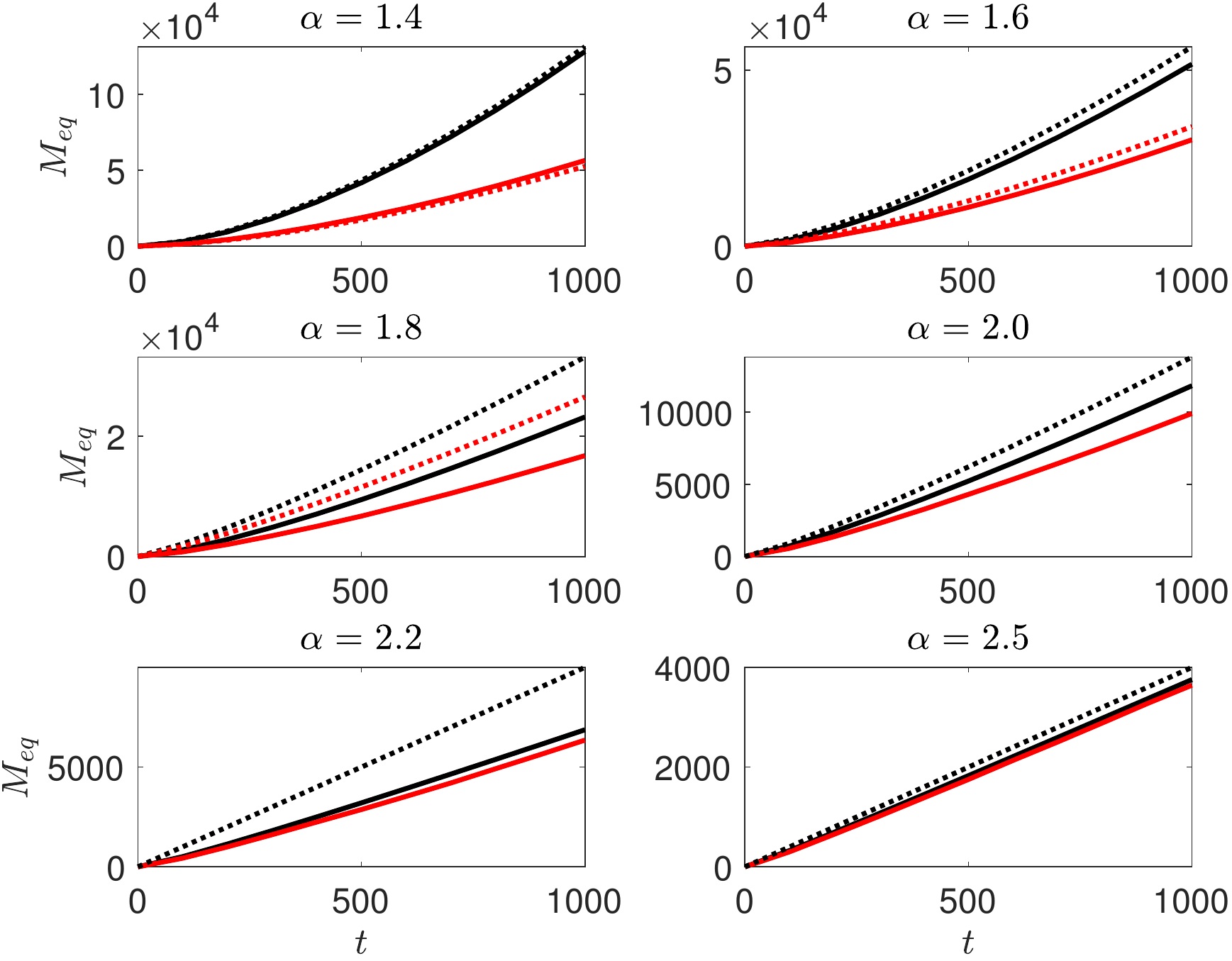}
\end{center}
\caption{Red: $M_{tr}$ (solid, computed by Monte Carlo simulation) and leading asymptotic approximation for $M_{tr}$ (dotted). 
Black: $M_{eq}$ (solid, computed analytically) and leading asymptotic approximation for $M_{eq}$ (dotted). For $\alpha \geq 2$, the
leading asymptotic approximations for $M_{tr}$ and $M_{eq}$ are identical, and are shown as black dotted curves.}
\label{fig:MC_F_ALPHA_TRANSITIONAL}
\end{figure}

\section{Application to free molecular flow in a planar channel} We studied
the flow of a rarefied gas in an infinite planar channel in \cite{BGT}. Assuming that length units are chosen so that the thickness of the channel is 1, we take the flow domain to be
$$
\{ (x_1,x_2,z) \in \R^3 ~:~ 0 < z< 1\}.
$$
We refer to $x=(x_1,x_2)$ as the ``horizontal" coordinates, and to $z$ as the ``vertical" coordinate.
As in \cite{BGT}, we consider
the projection of a particle trajectory into the $(x_1,x_2)$-plane.
Gas molecules are assumed not to interact with each other, to travel at constant
velocities in the interior of the channel, and to undergo random reflections at the walls, described by Mawell's boundary conditions \cite[pp.\ 118 ff]{Cercignani:112117} (the accomodation coefficient in \cite{BGT} is arbitrary, but we take it to be 1 here for simplicity). Specifically, a gas molecule that hits the 
lower wall $z=0$ re-emerges with a random velocity $(V,W)$,  $V \in \R^2$ and $W>0$, 
with density 
\begin{equation}
\label{eq:diffuse_reflection}
\frac{ e^{-|\nu|^2/c^2} }{ \pi c^2}  ~~ \frac{2\omega e^{-\omega^2/c^2}}{c^2} , ~~~ \nu \in \R^2, ~~ \omega>0. 
\end{equation}
The parameter $c$ equals
$\sqrt{2kT/m}$, with $T$ $=$ absolute temperature, $m$ $=$ mass per 
gas molecule, and $k$ $=$ Boltzmann constant; dimensionally, $c$ is a speed. The reflection law at the upper wall, $z=1$, is analogous, with the 
sign of $W$ reversed.

The time $\tau$ between a collision with a wall and the next collision with the opposite wall equals $1/|W|$.
The distribution function of $\tau$ is
$$
F(t) = 
P(\tau \leq t) = P \left( |W| \geq \frac{1}{t} \right)  =   \int_{1/t}^\infty \frac{2\omega e^{-\omega^2/c^2}}{c^2} d\omega = e^{-1/(ct)^2}.
$$
Note also that 
\begin{equation}
\label{BGT_F}
F(t) = 1 - \frac{1}{( c t)^2} + \frac{1}{2(ct)^4} - \ldots ~~~ \mbox{as $t \rightarrow \infty$.}
\end{equation}
The mean
of $\tau$ is
\begin{equation}
\label{BGT_taubar}
\taubar =  \int_0^\infty \left( 1-e^{-1/(ct)^2} \right) dt =  \frac{\sqrt{\pi}}{c}.
\end{equation}
The  two components of the horizontal velocity  $V$  are Gaussians with mean zero and variance $c^2/2$, independent
of each other and of $W$. Therefore
$$
v^2 = E(|V|^2) = E(V_1^2 + V_2^2) = c^2.
$$

From Proposition \ref{propode}, $M_{eq}$ can be expressed in terms of the error function erf
 and the exponential integral
$\mbox{Ei}$.
We omit the unwieldy formula, which however makes it easy to find asymptotic expansions of $M_{eq}$ 
to arbitrary accuracy. For instance, we find
\begin{equation}
\label{eq:BGT_M_eq}
M_{eq}(t) = \frac{2ct \log(ct)}{ \sqrt{\pi}} + \frac{1-\gamma}{\sqrt{\pi}}~c t  +  \frac{2}{3} + O \left( \frac{1}{t} \right)  ~~~ \mbox{as $t \rightarrow \infty$}, 
\end{equation}
where $\gamma \approx 0.577$ is the Euler-Mascheroni constant. Using eq.\ (\ref{eq:smoothed_M_tr}) 
and Lemma  \ref{lemma:anti_zafarty}, we also conclude
\begin{equation}
\label{eq:BGT_M_tr}
M_{tr}(t) =  \frac{2ct \log(ct)}{ \sqrt{\pi}} - \frac{1+ \gamma}{\sqrt{\pi}}~c t  + O \left( t^{1/2} (\log t)^{3/2} \right)
\end{equation}
for the transitional mean square displacement. 

It is interesting to compare (\ref{eq:BGT_M_eq}) and (\ref{eq:BGT_M_tr}) with 
the main result of \cite{BGT}. The notation in \cite{BGT} differs from that
used here in several ways. In particular, 
in the limit studied in \cite{BGT}, the channel width tends to zero while time tends to infinity.
However, 
 it is not difficult to show that the result of \cite{BGT}, translated
 into the notation used here, predicts that at time $t$, the 
distribution of the location of a particle starting at the origin at time $0$ will be approximately bivariate Gaussian with 
variance 
$$
\frac{ ct \log(ct)}{\sqrt{\pi}}.
$$
Comparing with (\ref{eq:BGT_M_eq}) we see the ``doubling effect" previously observed by others 
\cite{Balint_Dolgopyat,barkai1997levy,Zaburdaev2015levy} in which the variance of the long-time Gaussian approximation is half the leading-order MSD approximation. 

\section{Summary}

There is considerable interest in the MSD in L\'evy walks due to the importance of this quantity in a wide variety of application areas. The leading-order behavior of MSDs for L\'evy walks with power law distributions of the flight segment duration is asymptotically determined by the exponent of the power law, up to a constant of proportionality. Given a mean duration, the constant of proportionality is determined as well. However, the quality of the asymptotic approximation depends delicately on the exact form of the distribution. The dependence is especially sensitive for power laws with exponents near 2. At $\alpha=2$, the leading-order term of the MSD at time $t$ is proportional to $t\log t$, but the next-order term is usually proportional to $t$, with a constant of proportionality which depends on the entire distribution. For $\alpha$ near 2, the sub-leading-order terms can also be close to the leading-order terms for values of $t$ of physical interest.

We have derived exact, closed-form expressions, valid for all time, for the MSDs for particular choices of the power laws and for a power law perturbed by a logarithmic factor. These examples illustrate the dependence of the accuracy of the asymptotics on the entire distribution. We have also established robustness of the MSD asymptotics in the sense that power law bounds on the distribution functions imply power law bounds on the equilibrium MSDs. Finally, we have proved that for power laws with $\alpha$ between 1 and 2, the equilibrium and transitional MSDs are asymptotically proportional, with constant of proportionality $\alpha-1$.

\section*{Acknowledgment} The second author thanks the Courant Institute for hosting him as a Visiting Scholar.

\pagebreak


\end{document}